\documentclass[10pt, twocolumn]{IEEEtran}
\usepackage{color}
\usepackage{cite,xspace}
\usepackage{amsmath}
\usepackage{amsthm}
\usepackage{amsfonts}
\usepackage{amssymb}
\usepackage{algorithm}
\usepackage{algorithmic}
\usepackage{graphicx}
\usepackage{verbatim}

\newtheorem{theorem}{Theorem}
\newtheorem{lemma}{Lemma}
\newtheorem{remark}{Remark}

\newtheorem{example}{Example}

\newcommand{\beq}{\begin{equation}}
\newcommand{\eeq}{\end{equation}}
\newcommand{\beqnn}{\begin{equation*}}
\newcommand{\eeqnn}{\end{equation*}}
\newcommand{\beqy}{\begin{eqnarray}}
\newcommand{\eeqy}{\end{eqnarray}}
\newcommand{\beqynn}{\begin{eqnarray*}}
\newcommand{\eeqynn}{\end{eqnarray*}}
\newcommand{\bit}{\begin{itemize}}
\newcommand{\eit}{\end{itemize}}
\newcommand{\ben}{\begin{enumerate}}
\newcommand{\een}{\end{enumerate}}
\newcommand{\bex}{\begin{example}}
\newcommand{\eex}{\end{example}}

\newcommand{\balg}[1]{\begin{algorithm} \caption{#1}}
\newcommand{\ealg}{\end{algorithm}}

\newcommand{\balgc}{\begin{algorithmic}[1]}
\newcommand{\ealgc}{\end{algorithmic}}

\newcommand{\bary}{\begin{array}}
\newcommand{\eary}{\end{array}}
\newcommand{\bmx}{\begin{bmatrix}}
\newcommand{\emx}{\end{bmatrix}}
\newcommand{\bsmx}{\left[\begin{smallmatrix}}
\newcommand{\esmx}{\end{smallmatrix}\right]}
\newcommand{\bmxc}[1]{\left[\begin{array}{@{}#1@{}}}
\newcommand{\emxc}{\end{array}\right]}
\newcommand{\bcn}{\begin{center}}
\newcommand{\ecn}{\end{center}}

\newcommand{\Rbb}{{\mathbb{R}}}
\newcommand{\Zbb}{{\mathbb{Z}}}

\newcommand{\Rnbn}{\Rbb^{n \times n}}

\newcommand{\Rmbn}{\Rbb^{m \times n}}

\newcommand{\Zn}{\Zbb^{n}}

\newcommand{\sOSIC}{{\scriptscriptstyle \text{OSIC}}}
\newcommand{\sBSIC}{{\scriptscriptstyle \text{BSIC}}}

\newcommand{\A}{\boldsymbol{A}}

\newcommand{\I}{\boldsymbol{I}}

\newcommand{\Q}{\boldsymbol{Q}}
\newcommand{\R}{\boldsymbol{R}}

\newcommand{\e}{\boldsymbol{e}}

\renewcommand{\l}{\boldsymbol{\ell}}

\renewcommand{\u}{\boldsymbol{u}}
\renewcommand{\v}{\boldsymbol{v}}

\newcommand{\x}{{\boldsymbol{x}}}
\newcommand{\y}{{\boldsymbol{y}}}

\newcommand{\0}{{\boldsymbol{0}}}

\newcommand{\bbR}{{\bar{\R}}}

\newcommand{\hx}{{\hat{x}}}

\newcommand{\hbx}{{\hat{\x}}}

\newcommand{\iid}{i.i.d.\xspace}

\begin{document}
\title{Closed-Form Word Error Rate Analysis for Successive Interference Cancellation Decoders}

 \author{Jinming~Wen, Keyu Wu, \\  Chintha Tellambura, \IEEEmembership{Fellow, IEEE},  and Pingzhi Fan, \IEEEmembership{Fellow, IEEE}
 \thanks{This work was presented in part at the 2017 IEEE International Conference on Communications (ICC), Paris, France.}
\thanks{J.~Wen is with  the College of Information Science and Technology, and the College of Cyber Security, Jinan University, Guangzhou 510632, China (e-mail: jinming.wen@mail.mcgill.ca).}
\thanks{K.~Wu is with the School of Electronic Science,
	National University of Defense Technology, Changsha 410073, China (e-mail: keyuwu@nudt.edu.cn).}
 \thanks{C.~Tellambura is with  the Department of Electrical and Computer Engineering, University of Alberta, Edmonton T6G 2V4, Canada (e-mail: chintha@ece.ualberta.ca)}
 \thanks{P.~Fan is  with the Institute of Mobile Communications, Southwest
Jiaotong University, Chengdu 610031, China (e-mail: p.fan@ieee.org)}

\thanks{The work of Pingzhi Fan was supported by NSFC-NRF project under grant No.61661146003/NRF2016NRF-NSFC001-089, and the 111 Project under Grant 111-2-14.
This work was also partially supported by National Natural Science Foundation of China (No. 11871248),
``the Fundamental Research Funds for the Central Universities'' (No. 21618329)
and the postdoc research fellowship from Fonds de Recherche Nature et Technologies.}

 %\thanks{Manuscript received; revised .}
}

\maketitle

\begin{abstract}
We consider the estimation of an integer  vector $\hbx\in \mathbb{Z}^n$ from the linear observation
$\y=\A\hbx+\v$, where $\A\in\mathbb{R}^{m\times n}$ is a random matrix with independent and identically distributed (i.i.d.) standard Gaussian  $\mathcal{N}(0,1)$ entries, and $\v\in \mathbb{R}^m$ is a noise vector with  i.i.d. $\mathcal{N}(0,\sigma^2 )$ entries with given $\sigma$.  In digital communications, $\hbx$ is typically  uniformly distributed over an  $n$-dimensional  box $\mathcal{B}$. For this estimation problem, successive interference cancellation (SIC) decoders are popular   due to their  low  complexity, and a detailed analysis of  their word error rates (WERs) is highly useful.  In this paper, we derive closed-form WER  expressions for two  cases: (1)  $\hbx\in \mathbb{Z}^n$  is fixed and (2)
$\hbx$ is uniformly distributed over $\mathcal{B}$. We also investigate some of their properties in detail and show that they agree closely with   simulated word error probabilities.
\end{abstract}


\section{Introduction}
\label{s:introduction}

\subsection{Motivation}
\IEEEPARstart{I}{nteger}  parameter estimation \cite{HasB98} in linear models finds many applications such as Global Positioning System (GPS), cryptography, digital communications, code division multiple access  and others.  The prototype problem is to estimate (detect)  an integer  vector $\hbx\in\mathbb{Z}^n$ from the  linear model:
\begin{equation}
\label{e:model}
\y=\A\hbx+\v,  \quad \v \sim \mathcal{N}(\0,\sigma^2\I),
\end{equation}
where $\y\in \mathbb{R}^m$ is an observation vector,
$\A\in\mathbb{R}^{m\times n}$ is a random matrix with \iid  standard Gaussian $\mathcal{N}(0,1)$ entries
and $\v\in \mathbb{R}^m$ is a Gaussian  noise vector
$\mathcal{N}(\boldsymbol{0},\sigma^2 \I)$ with variance  $\sigma^2$ of each entry.

The maximum-likelihood (ML)  estimator of   $\hbx$ is the solution of a simple least-squares problem if the integer constraint is relaxed (e.g., $\hbx \in \mathbb R^n $). However, such relaxation is not highly accurate. Thus, the exact ML estimator of $ \hbx $ is given by the solution of the  following integer least-squares (ILS) problem \cite{HasB98} \cite{HasV05}:
\beq
\label{e:ILS}
\min_{\x\in{\mathbb{Z}}^n}\|\y-\A\x\|_2.
\eeq
Because solving  \eqref{e:ILS} is equivalent to finding the  closest point to $\y$
in the lattice $\{\A\x: \x\in \Zn\}$,
problem \eqref{e:ILS} is also referred to as the  closest-point problem in cryptography  \cite{AgrEVZ02}. In terms of complexity, this problem is Non-deterministic Polynomial (NP)-hard.

In digital  communication links, prior to transmission,    data bits are mapped to  a fixed set of modulation symbols (signal constellation). For example, Section IV discusses  $M$-ary pulse amplitude modulation (PAM) constellation, which consists of $M$ integers. Thus, with  $M$-ary PAM,     the entries of $\hbx $ are selected from the  fixed constellation of integers.
The signal constellations are also subject to the average power constraints. Thus,  the parameter vector $\hbx$  satisfies a
box constraint \cite{DamGC03, SuW05, ChaH08, BreC11, ParKL16}, i.e.,
\beq
\label{e:box}
\hbx \in {\cal B}:=\{ \x: \l \leq \x \leq \u, \ \x,  \l, \u\in \Zn\}.
\eeq
In practical systems,  all  signal constellation points  are equally likely, which  is equivalent to  $\hbx$ being  uniformly distributed over $\mathcal{B}$,
see, e.g., \cite{JalO05, WenC17}. Thus, the box constraint (3) can be incorporated in  \eqref{e:ILS}, which yields the so-called
 box-constrained integer least-squares (BILS) problem:
\beq
\label{e:BILS}
\min_{\x\in\mathcal{B}}\|\y-\A\x\|_2.
\eeq

Problems \eqref{e:ILS} and \eqref{e:BILS} can be optimally  solved by a sphere decoder
(see  \cite{HasV05} and \cite{ChaH08}), which  consists of pre-processing and search stages.
For example, one can pre-process matrix $ \A$ by using the Lenstra-Lenstra-Lov\'{a}sz (LLL) algorithm  \cite{LenLL82}, which reduces $ \A$ to a nearly orthogonal lattice basis, which improves the efficiency of the search stage.  Other pre-processing strategies include  Vertical-Bell Labs layered Space Time algorithm (V-BLAST) \cite{FosGVW99},  Sorted QR Decomposition (SQRD) \cite{WubBRKK01}
and their variants \cite{SuW05, ChaH08, BreC11}.
%Note that in addition to the V-BLAST and SQRD, one can also use the native lattice reduction \cite{TahK10}
%to reduce \eqref{e:BILS}.
Perhaps, the most  frequently utilized discrete search algorithms for   \eqref{e:ILS} or \eqref{e:BILS} are  the Schnorr-Euchner search algorithm  \cite{SchE94}  and  its variants
 \cite{AgrEVZ02,DamGC03, CuiT05,CuiT05b, GhaT08, CuiHT13, WenZMC16}.

It  has respectively been shown in \cite{Mic01} and \cite{JalO05} that
\eqref{e:ILS} and  \eqref{e:BILS} are NP-hard problems;
hence, for many applications,  suboptimal algorithms are common.
A  popular  one  for solving  \eqref{e:ILS} is  the
ordinary successive interference cancellation
(OSIC) decoder,  which is actually Babai's  nearest plane algorithm \cite{Bab86}.
It  can also be adapted to form  a box-constrained SIC (BSIC) decoder,  a suboptimal algorithm for \eqref{e:BILS}. Interestingly, since the Schnorr-Euchner algorithm is a depth-first search, the first valid solution found by it, is in fact the   OSIC  decoder solution, also called    Babai point \cite{AgrEVZ02, ChaWX13}. Similarly, the initial solution of
the Schnorr-Euchner decoder of  \eqref{e:BILS} is  the BSIC decoder solution, which is
a box-constrained Babai point \cite{AgrEVZ02, DamGC03, ChaH08, WenC17}.

Analyzing the performance of decoders helps to design and  characterize wireless communication links \cite{PraV04,ZanCW05, LiuK11, MenK15,SonYCA15, QianM17,BaoMKXZ17}. The most common  decoder performance measures  involve the error probability of the decoding process. Specifically,   we utilize the error probability that the output of the decoder is not  equal to the true integer vector $\hbx$, which is called  word error rate (WER).
The probability of correct detection is called the success probability  \cite{HasB98,ChaWX13, WenTB16, WenC17}.

The WER characterization of both OSIC and BSIC decoders is useful  \cite{ChaWX13,WenC17}.
Indeed, with  OSIC decoder solving  \eqref{e:ILS}
or a BSIC decoder solving  \eqref{e:BILS},
their WERs, respectively denoted by, $P^{\sOSIC}_e$ and $P^{\sBSIC}_e$,
serve as critical quality parameters. For instance, a suitable threshold can be setup a priori -- if the WER is below is  threshold -- to indicate that
the decoder can be used with confidence.
In this case, the additional effort of optimally solving  the ILS \eqref{e:ILS} or the BILS \eqref{e:BILS} yields diminishing returns.
However, if $P^{\sOSIC}_e$ or $P^{\sBSIC}_e$  is  above the threshold, then more  accurate  decoders, such as
a sphere decoder (ML  estimator), should be used.
Even if one intends to solve the ILS \eqref{e:ILS} or the BILS \eqref{e:BILS} for ML estimator of $\hbx$, it is still of vital importance to compute $P^{\sOSIC}_e$ or $P^{\sBSIC}_e$ since they are often used to approximate their WER.

\subsection{Contributions}
Closed-form expressions for $P^{\sOSIC}_e$ and $P^{\sBSIC}_e$
have respectively been given in \cite{ChaWX13} and \cite{WenC17} when $\A$ in \eqref{e:model} is
deterministic.
Moreover, closed-form WER  expressions for zero-forcing  and BSIC decoders have been derived for when  $\hat{\x}$ is a  fixed integer vector  and for when  $\hat{\x}$ is  uniformly distributed over $\mathcal{B}$
for deterministic $\A$~\cite{WenCT17}. The relationship between  WERs of zero-forcing and BSIC decoders
was  also investigated in \cite{WenCT17}.
However, all of these formulas are for deterministic $\A$.
To the best of our  knowledge, for random  $\A$,   the WER analysis for SIC decoders has been lacking.
This paper fills this gap and  derives closed-form WER  expressions for both OSIC and BSIC cases.
Specifically, the contributions can be summarized as follows:
\begin{itemize}
\item We derive a closed-form WER  expression $P^{\sOSIC}_e$ for the  SIC decoder
when  $\hbx$ is  a fixed integer vector,
and investigate some of the  properties of $P^{\sOSIC}_e$.  In particular, we rigorously show that $P^{\sOSIC}_e$
tends to 0 when $\sigma^2$, which is the noise variance, tends to 0,
and quantify the gap of $P^{\sOSIC}_e$ for two sizes $n_1$ and $n_2$
(Section \ref{s:OILS})\footnote{This paper   was presented in part at  2017 IEEE International Conference on Communications (ICC)~\cite{WenWT17}.}.

\item We derive a  closed-form WER  expression $P^{\sBSIC}_e$ for  BSIC decoder when  $\hbx$ is
uniformly distributed over $\mathcal{B}$, and investigate some of its properties.
In particular, we rigorously show that $P^{\sBSIC}_e$
tends to 0 when $\sigma$ tends to 0,
and quantify the gap of $P^{\sBSIC}_e$ for two sizes $n_1$ and $n_2$ (Section \ref{s:BILS}).
\item We study the relationship between $P^{\sOSIC}_e$ and $P^{\sBSIC}_e$.
More precisely, we show that $P^{\sBSIC}_e\leq P^{\sOSIC}_e$
and they converge to one value as  noise variance $\sigma^2$ tends to 0 (Section \ref{s:relation}).
\end{itemize}

\subsection{Comparison with existing work}
Many works have  theoretically analyzed the performance of some commonly used decoders \cite{PraV04,ZanCW05, LiuK11}.
Although our  closed-form WER analysis   has some connections with those in
\cite{PraV04,ZanCW05, LiuK11}, there are main differences between them.
More specifically:
\begin{enumerate}
\item  Our closed-form expressions (see eq. \eqref{e:pb} and eq. \eqref{e:pbbox} in Sec.
\ref{s:OILS} and \ref{s:BILS}) for the WER of OSIC and BSIC decoders are simpler
and more concise than \cite[Theorem 1]{PraV04} (note that \cite[eq. (14)]{PraV04} is more complicated than eq. \eqref{e:R}), \cite[eq. (18)]{ZanCW05} and \cite[Theorem 1]{LiuK11}.
Because of this  simplicity, we can theoretically characterize the gap
of the WER corresponding to two different dimensions of $\A$
(Theorems \ref{t:propOSIC2} and \ref{t:propBSIC2}).
However, we do not find similar results in \cite{PraV04,ZanCW05, LiuK11}.
\item Another common difference between this paper and \cite{PraV04,ZanCW05, LiuK11} is the techniques for the
WER analysis. Our main techniques for the WER analysis are the distribution of the triangular factor of the QR factorization of the random matrix $\A$, chain rule, random available transformation and the computational formulas of OSIC and BSIC which are simple and clear.
The main techniques of the joint error probability analysis in \cite{PraV04} are the distribution of the triangular factor of the QR factorization of the random matrix $\A$, chain rule and a result from \cite{SimA00}.       Reference
\cite{ZanCW05} mainly uses the total probability theorem and some approximation techniques.
The main technique of \cite{LiuK11} is based on some analysis on n-PSK modulation.
There are also some other differences between them, outlined below:

\item Another difference between this paper and \cite{PraV04} is that
our WER analysis is valid for any box $\mathcal{B}$, while \cite{PraV04} assumes that $\mathcal{B}$
is a cube with the edge length $2^{2z}$, where $z$ is a positive integer. Since in some applications, such as when the constellations are 4-QAM,  the edge length of $\mathcal{B}$ does not satisfy $2^{2z}$ for a positive integer $z$, and the analysis of WER over an arbitrary  box $\mathcal{B}$ is still needed.

\item Different from our paper which analyzes the WER of OSIC and BSIC decoders,
\cite{ZanCW05} investigates the bit error rates of both minimum mean square error (MMSE)-non-SIC
and MMSE-SIC.
From eq.\eqref{e:OB}-eq.\eqref{e:BB} and  \cite[eq.s (2-7)]{ZanCW05},
we can see that these two papers study the error performance of different decoders.

\item There are three additional differences between this paper and \cite{LiuK11}:
firstly, our analysis is valid for any box $\mathcal{B}$, which is different from \cite{LiuK11} that
assumes $\mathcal{B}$ is transformed from $n$-PSK modulators.
Secondly, the WER in this paper refers to the probability that a decoder does not successfully detect
$\hbx$, which is different from the symbol error probability in \cite{LiuK11} (please see
\cite[eq.s (18) and (22)]{LiuK11}).
Thirdly, we give closed-form expressions for the \emph{exact}
WER of OSIC and BSIC decoders, whereas \cite{LiuK11} proposes an  \emph{approximation}
of the symbol error probability of  multiple-input and multiple-output-MMSE-SIC decoders.
\end{enumerate}

The rest of the paper is organized as follows.
In Section \ref{s:QR}, we introduce the computational details of OSIC and BSIC decoders.
In Section \ref{s:OILS}, we develop a closed-form expression for $P^{\sOSIC}_e$  and investigate its properties.
In Section \ref{s:BILS}, we develop  closed-form  $P^{\sBSIC}_e$
and  study its properties.
The relationship between $P^{\sOSIC}_e$ and $P^{\sBSIC}_e$ is analyzed in Section \ref{s:relation}.
Numerical simulations to verify the  derived formulas are presented  in Section \ref{s:Sim}.
Finally, we summarize and discuss our results  in Section \ref{s:sum}.

\textbf{Notation:} For  a vector $\x$,
 $\lfloor \x\rceil$  denotes its nearest integer vector,
i.e., each entry of $\x$ is rounded to its nearest integer
(if there is a tie, rounding is downward),
and  $x_i$ denotes the $i$-th element of $\x$.
Let $a_{ij}$ be the element of matrix $\A$ at row  $i$ and column  $j$.
Let $P^{\sOSIC}_e$ and $P^{\sBSIC}_e$ respectively denote the WER of
the SIC and BSIC decoders

\section{OSIC and BSIC decoders}
\label{s:QR}
In this section, we briefly   introduce the computational details  of OSIC and BSIC decoders.

Suppose that $\A$ in \eqref{e:model} has the following thin QR factorization \cite[p.230]{GolV13}:
\beq
\label{e:qr}
\A=\Q\R,
\eeq
where  $\Q\in \Rmbn$ is an orthonormal matrix and $\R\in \Rnbn$ is an upper triangular matrix.
Let $\bar{\y}=\Q^T\y$ and $\bar{\v}=\Q^T\v$.
Since $\v \sim \mathcal{N}(\0,\sigma^2\I)$,
$\bar{\v} \sim \mathcal{N}(\0,\sigma^2\I)$. By \eqref{e:qr}, eq. \eqref{e:model} can be transformed to
\beq
\label{e:modeltransf}
\bar{\y}=\R\hbx+\bar{\v}, \quad \bar{\v} \sim \mathcal{N}(\0,\sigma^2\I).
\eeq

The output of the OSIC decoder  $\x^\sOSIC\in \Zbb^n$ is  computed as follows \cite{Bab86}:
\beq
\label{e:OB}
 c_{i}^\sOSIC=(\bar{y}_{i}-\sum_{j=i+1}^nr_{ij}x_j^\sOSIC)/r_{ii},\quad
 x_i^\sOSIC=\lfloor c_i^\sOSIC\rceil
\eeq
for $i=n, n-1, \ldots, 1$, where $\sum_{n+1}^n r_{nj}x_j^\sOSIC=0$.

By modifying the Babai nearest plane algorithm  \cite{Bab86} with taking
the constrained box into account, one can get a BSIC decoder (see, e.g., \cite{WenC17}).
The output of BSIC decoder   $\x^\sBSIC\in \mathcal{B}$ can be computed via
\beq
\label{e:BB}
\begin{split}
 c_{i}^\sBSIC&=(\bar{y}_{i}-\sum_{j=i+1}^nr_{ij}x_j^\sBSIC)/r_{ii}, \\ \ \
 x_i^\sBSIC&=
\begin{cases}
\ell_i, & \mbox{ if }\   \lfloor c_i^\sBSIC\rceil\leq \ell_i\\
\lfloor c_i^\sBSIC\rceil, & \mbox{ if }\    \ell_i<\lfloor c_i^\sBSIC\rceil< u_i\\
u_i, & \mbox{ if }\    \lfloor c_i^\sBSIC\rceil \geq u_i
\end{cases}
\end{split}
\eeq
for $i=n, n-1, \ldots, 1$, where $\sum_{n+1}^nr_{nj}x_j^\sBSIC =0$.

\section{WER for OSIC Decoders}
\label{s:OILS}

In this section, we derive  closed-form  $P^{\sOSIC}_e$  and investigate its properties.

\subsection{WER for OSIC Decoders}

This subsection derives the $P^{\sOSIC}_e$ expression.
To this end, we introduce two lemmas which are
	needed for the one dimensional case and for characterizing the distribution
of the entries of $\R$ in \eqref{e:qr}. We begin by  introducing the first lemma.

\begin{lemma}
\label{l:probk} Consider the  following scalar  linear model:
\beq
\label{e:modelk}
\bar{y}=r\hat{x}+\bar{v}, \quad \bar{v} \sim \mathcal{N}(0,\sigma^2),
\eeq
where $\hat{x}\in \mathbb{Z}$ is a fixed unknown parameter number,
$\bar{v}\in \mathbb{R}$ is a  Gaussian $\mathcal{N}(0,\sigma^2)$ noise term,
and $r^2>0$, which is independent with $\bar{v}$, is a chi-square  $\chi^2_{k}$ random variable with $ k>0$ degrees of freedom.
Let $x=\lfloor\bar{y}/r\rceil$, then
\beq
\label{e:R}
P_k= \Pr(x=\hat{x}) = C_k \int_0^{\arctan(1/(2\sigma))}\cos^{k-1}(\theta)d\theta
\eeq
where
\beq
\label{e:Ck}
C_k=\frac{2\Gamma((k+1)/2)}{\sqrt{\pi}\Gamma(k/2)}.
\eeq
\end{lemma}

\begin{proof}
See Appendix \ref{s:probkproof}.
\end{proof}

To derive the main theorem for $P^{\sOSIC}_e$,
we  introduce the following lemma  from \cite[P. 99]{Mui82}.

\begin{lemma}
\label{l:distribution}
Let the  entries of $\A   \in \mathbb{R}^{m \times n} $ be \iid
Gaussian  $\mathcal{N}(0,1)$ terms.
Then all $r_{ij}, 1\leq i\leq j\leq n$, are independent.
Moreover, $r_{ii}^2\sim\chi^2_{m-i+1}$ and $r_{ij}\sim \mathcal{N}(0,1)$
for $1\leq i < j\leq n$.
\end{lemma}

Based on Lemmas \ref{l:probk} and \ref{l:distribution}, the following theorem
for $P_e^{\sOSIC}$ can be obtained.

\begin{theorem}
\label{t:prob}
The word error rate $P_e^{\sOSIC}$ of OSIC decoder (see \eqref{e:OB}) satisfies
\beq
\label{e:pb}
P_e^{\sOSIC}\equiv\Pr(\x^{\sOSIC}\neq\hbx)=1-\prod_{i=1}^n P_{m-i+1},
\eeq
where $P_i$ is defined in \eqref{e:R}.
\end{theorem}

To prove Theorem \ref{t:prob}, we first use the chain rule of conditional
probabilities to transform $1-P_e^{\sOSIC}$ to the product of $n$ terms with each of them
representing a one-dimensional conditional success probability.
We use Lemma \ref{l:probk} to compute each term and finally obtain \eqref{e:pb}.
The detail is in the proof below.

\begin{proof}
Let
\[
P_s^{\sOSIC}=\Pr(\x^{\sOSIC}=\hbx)=1-P_e^{\sOSIC},
\]
then by the chain rule of conditional probabilities, we have
\begin{align*}
P_s^{\sOSIC} &  =\Pr\left(\bigcap_{i=1}^n(x_i^\sOSIC=\hx_i)\right) =\Pr(x_n^\sOSIC=\hx_{n}) \nonumber \\
& \quad\times \prod_{i=1}^{n-1}\Pr\left((x_i^\sOSIC=\hx_i)|\bigcap_{j=i+1}^n(x_{j}^\sOSIC=\hx_{j})\right).
\end{align*}
Thus, to show \eqref{e:pb}, we show
\begin{align}
 \label{e:SuccPron}
&\Pr(x_n^\sOSIC=\hx_{n})=P_{m-n+1},\\
&\Pr\left((x_i^\sOSIC=\hx_i)|\bigcap_{j=i+1}^n(x_{j}^\sOSIC=\hx_{j})\right)=P_{m-i+1},
\label{e:SuccProi}
\end{align}
for $i=n-1,n-2,\ldots,1$.

By \eqref{e:modeltransf},
\beq
\label{e:disn}
\bar{y}_n=r_{nn}\hat{x}_n+\bar{v}_n, \quad \bar{v}_n \sim \mathcal{N}(0,\sigma^2),
\eeq
and for $i=n-1,\ldots, 1$,
\beq
\label{e:disi1}
\bar{y}_i-\sum_{j=i+1}^n r_{ij}\hx_j=r_{ii}\hat{x}_i+\bar{v}_i, \quad \bar{v}_i \sim \mathcal{N}(0,\sigma^2).
\eeq
Clearly, if $ x_{i+1}^\sOSIC=\hx_{i+1}, \cdots, x_n^\sOSIC=\hx_{n}$,
by \eqref{e:OB}, \eqref{e:disn} and \eqref{e:disi1}, we can see that, for $i=n,\ldots, 1$,
\beq
\label{e:disi}
r_{ii}\,c_i^\sOSIC=r_{ii}\hat{x}_i+\bar{v}_i, \quad \bar{v}_i \sim \mathcal{N}(0,\sigma^2).
\eeq

By Lemma \ref{l:distribution},
\[
r_{ii}^2\sim\chi^2_{m-i+1}, \quad i=n, n-1,\ldots, 1.
\]
Thus, by \eqref{e:disi} and Lemma \ref{l:probk}, we can see that both \eqref{e:SuccPron} and \eqref{e:SuccProi} hold. Hence, the theorem holds.
\end{proof}

\begin{remark}
By \eqref{e:Ck},
\begin{align}
\label{e:Ckequivalent}
\prod_{i=1}^nC_{m-i+1}&=
\prod_{i=1}^n\left(\frac{2}{\sqrt{\pi}}\frac{\Gamma((m-i+2)/2)}{\Gamma((m-i+1)/2)} \right)\nonumber\\
&=\left(\frac{2}{\sqrt{\pi}}\right)^n\frac{\Gamma((m+1)/2)}{\Gamma((m-n+1)/2)}.
\end{align}
Thus, by \eqref{e:R}, eq. \eqref{e:pb} can be rewritten as
\beq
\label{e:pb2}
P_e^{\sOSIC}=1-\alpha\prod_{i=1}^n\int_0^{\arctan(1/(2\sigma))}\cos^{m-i}(\theta)d\theta,
\eeq
where
\[
\alpha=\left(\frac{2}{\sqrt{\pi}}\right)^n\frac{\Gamma((m+1)/2)}{\Gamma((m-n+1)/2)}.
\]
Note that \eqref{e:pb2} gives a more efficient way than \eqref{e:pb} for
computing $P_e^{\sOSIC}$ since computing $\alpha$ is slightly more efficient than
computing $\prod_{i=1}^nC_{m-i+1}$.
\end{remark}

\begin{remark}
In digital  communications, matrix $\A$ is often square. That is  $m=n$. Thus, it is useful to  simplify
$P_e^{\sOSIC}$ in \eqref{e:pb2} under this condition.
Since $\Gamma(1/2)=\sqrt{\pi}$, when $m=n$, we have
\begin{align*}
\alpha=&\left(\frac{2}{\sqrt{\pi}}\right)^n\frac{\Gamma((m+1)/2)}{\Gamma((m-n+1)/2)}
=\left(\frac{2}{\sqrt{\pi}}\right)^n\frac{\Gamma((n+1)/2)}{\sqrt{\pi}}\\
=&\frac{2^n\Gamma((n+1)/2)}{\sqrt{\pi^{n+1}}}
\end{align*}
and
\begin{align*}
&\,\prod_{i=1}^n\int_0^{\arctan(1/(2\sigma))}\cos^{m-i}(\theta)d\theta\\
=&\prod_{i=1}^n\int_0^{\arctan(1/(2\sigma))}\cos^{n-i}(\theta)d\theta\\
=&\prod_{j=n}^1\int_0^{\arctan(1/(2\sigma))}\cos^{j-1}(\theta)d\theta\\
=&\prod_{j=1}^n\int_0^{\arctan(1/(2\sigma))}\cos^{j-1}(\theta)d\theta,
\end{align*}
where the second equality follows form the transformation that $j=n-i+1$.
Hence, when $m=n$, \eqref{e:pb2} can be rewritten as
\begin{align*}
%\label{e:pbsquare}
P_e^{\sOSIC}=1-
\frac{2^n\Gamma((n+1)/2)}{\sqrt{\pi^{n+1}}}\prod_{i=1}^n\int_0^{\arctan(1/(2\sigma))}\cos^{i-1}(\theta)d\theta.
\end{align*}
\end{remark}

\subsection{Properties of OSIC Decoders}

We now investigate some properties of $P_e^\sOSIC$.
We begin with presenting the following important lemma,
which can be used to show that  $P_e^\sOSIC$ tends to 0 if noise level $\sigma$ tends to 0
for the one dimensional case.

\begin{lemma}
\label{l:integralk}
For any integer $k$, it holds that
\beq
\label{e:integralk}
\int_0^{\pi/2}\cos^{k-1}(\theta)d\theta=\frac{1}{C_k}.
\eeq
\end{lemma}

Lemma \ref{l:integralk} can be obtained from \cite[(24)]{SonC12}.

\begin{remark}
Since
\[
\lim_{\sigma\rightarrow0}\arctan\left(\frac{1}{2\sigma}\right)=\frac{\pi}{2},
\]
by \eqref{e:R} and \eqref{e:integralk}, one can easily see that,
for any integer $k$, we have
\beq
\label{e:Rklimit}
\lim_{\sigma\rightarrow0}P_k=1.
\eeq
\end{remark}

By \eqref{e:Rklimit}, we have the following result.

\begin{theorem}
\label{t:propOSIC}
The WER $P_e^\sOSIC$ (see \eqref{e:pb}) of OSIC decoders is
an increasing function of $\sigma$ and $n$. Moreover, it satisfies
\beq
\label{e:propOSIC}
\lim_{\sigma\rightarrow0}P_e^\sOSIC=0.
\eeq
\end{theorem}

\begin{proof}
By  \eqref{e:integralk}, one can easily see that for any fixed $\sigma$, we have
\beq
\label{e:Rklimit2}
\int_0^{\arctan(1/(2\sigma))}\cos^{k-1}(\theta)d\theta<\frac{1}{C_k},
\eeq
which combing with \eqref{e:R} implies that $P_k<1$ for any fixed $\sigma$.
Thus,  by \eqref{e:pb}, $P_e^\sOSIC$ is an increasing function of  $n$
for any fixed $\sigma$.
One can easily show that $P_e^\sOSIC$ is an increasing function of  $\sigma$
for any fixed $n$, thus, the first part of the result holds.

By \eqref{e:pb} and \eqref{e:Rklimit}, we have
\begin{align*}
\lim_{\sigma\rightarrow0}P_e^{\sOSIC}
=&1-\lim_{\sigma\rightarrow0}\prod_{i=1}^nP_{m-i+1}\\
=&1-\prod_{i=1}^n\lim_{\sigma\rightarrow0}P_{m-i+1}=0.
\end{align*}
Thus, eq. \eqref{e:propOSIC} holds.
\end{proof}

Note that Theorem \ref{t:propOSIC} also holds for deterministic $\A$.
More details can be found in \cite[Corollary 2]{WenC17}.

In many applications, matrix $\A$ is a square matrix.
For ease of notation,  let  the WER of OSIC decoder be $P_e^\sOSIC(n)$ when  matrix  $\A$ is $n\times n$.  The  following results can be directly obtained from \eqref{e:pb}.

\begin{theorem}
\label{t:propOSIC2}
Let $n_1<n_2$ be two integers, then $P_e^\sOSIC(n_1)$ and $P_e^\sOSIC(n_2)$,
which are respectively the WER of OSIC decoders for sizes $n_1$ and $n_2$ satisfy
\begin{equation}
\label{e:SRR}
\frac{1-P_e^\sOSIC(n_2)}{1-P_e^\sOSIC(n_1)}=\prod_{k=n_1+1}^{n_2} P_k.
\end{equation}
\end{theorem}

Theorem \ref{t:propOSIC2} quantifies the gap between two $P_e^\sOSIC$ for  two different sizes. Specifically, if noise level  $\sigma$ converges to  0, then by \eqref{e:Rklimit},
$P_k$ is close to 1 for any integer $k$.
Thus, eq. \eqref{e:SRR} indicates that when  noise level   $\sigma$ converges  to 0,
the difference between $1-P_e^\sOSIC(n_1)$ and $1-P_e^\sOSIC(n_2)$ is  small,
implying that the gap between $P_e^\sOSIC(n_1)$ and $P_e^\sOSIC(n_2)$ is very small
as long as noise level  $\sigma$ is near  0.
For more details, see the numerical experiments in Section \ref{s:Sim}.

\section{WER for BSIC Decoders}
\label{s:BILS}
As mentioned before, for digital wireless communications and other applications,  $\hbx$  is uniformly distributed over $\mathcal{B}$. For  this condition,  we analyze the WER of BSIC decoder.

\subsection{WER for BSIC Decoders}

To derive  closed-form  $P^{\sBSIC}_e$, we first  introduce the following useful
lemma, which analyzes the WER for one dimensional case.

\begin{lemma}
\label{l:probkbox2}
Suppose that we have the scale linear model \eqref{e:modelk},
where $\hat{x}\in \mathbb{Z}$ is uniformly distributed on $[\ell,u]$,
$\bar{v}\in \mathbb{R}$ is a noise number following the Gaussian distribution $\mathcal{N}(0,\sigma^2)$,
and $r^2>0$, which is independent with $\bar{v}$,
follows central chi-square distribution $\chi_k^2$ with $k>0$ degree of freedom.
Let
\beq
\label{e:BBk}
\begin{split}
 x&=
\begin{cases}
\ell, & \mbox{ if }\   \lfloor\bar{y}/r\rceil\leq \ell\\
\lfloor\bar{y}/r\rceil, & \mbox{ if }\    \ell<\lfloor\bar{y}/r\rceil< u\\
u, & \mbox{ if }\    \lfloor\bar{y}/r\rceil \geq u
\end{cases}.
\end{split}
\eeq
Then $x$ satisfies
\begin{align}
\label{e:pkbox2}
\Pr(x=\hat{x})=\bar{P}_k(u-\ell),
\end{align}
where for $\eta>0$,
\begin{align}
\label{e:Rbar}
\bar{P}_k(\eta)
=\frac{C_k}{\eta+1}\left(\frac{1}{C_k}+\eta\int_{0}^{\arctan(1/2\sigma)}
\cos^{k-1}(\theta)d\theta\right)
\end{align}
with $C_k$ being defined in \eqref{e:Ck}.
\end{lemma}

\begin{proof}
See Appendix \ref{s:probkproofbox2}.
\end{proof}

By using  Lemmas \ref{l:distribution} and \ref{l:probkbox2},
we have the following theorem for $P_e^{\sBSIC}$.

\begin{theorem}
\label{t:probbox}
Suppose that $\hbx$ in \eqref{e:model} is uniformly distributed over the constraint box $\mathcal{B}$ (see \eqref{e:box}), and $\hbx$ and ${\v}$ are independent.
Then, the word error rate $P_e^{\sBSIC}$ of BSIC decoder (see \eqref{e:BB}) satisfies
\beq
\label{e:pbbox}
P_e^{\sBSIC}\equiv\Pr(\x^{\sBSIC}\neq\hbx)=1-\prod_{i=1}^n \bar{P}_{m-i+1}(u_i-\ell_i),
\eeq
where $\bar{P}_{m-i+1}(u_i-\ell_i)$ is defined in \eqref{e:Rbar}.
\end{theorem}

Since $\hbx$ is uniformly distributed over  $\mathcal{B}$,
$\hat{x}_i$ is uniformly distributed on $[\ell_i,u_i]$ for $1\leq i\leq n$.
Theorem \ref{t:probbox} can be proved by using more or less the same techniques as
that for Theorem \ref{t:prob}, thus we omit its proof.

\begin{remark}
Similar to the ordinary case, by  \eqref{e:Ckequivalent} and \eqref{e:Rbar},
eq. \eqref{e:pbbox} can be rewritten as
\begin{align}
\label{e:pbbox2}
P_e^{\sBSIC}=1-&\beta\prod_{i=1}^n\hat{P}_i,
\end{align}
where
\[
\beta=\left(\frac{2}{\sqrt{\pi}}\right)^n\frac{\Gamma(m+1)/2}{\Gamma(m-n+1)/2}\prod_{i=1}^n
\frac{1}{(u_i-\ell_i+1)}
\]
and
\[
\hat{P}_i=\frac{1}{C_{m-i+1}}+(u_i-\ell_i)\int_{0}^{\arctan(1/2\sigma)}\cos^{m-i}(\theta)d\theta
\]
with $C_{m-i+1}$ being defined in \eqref{e:Ck}.
Clearly, $P_e^{\sBSIC}$ computed by \eqref{e:pbbox2} is more efficient than that via \eqref{e:pbbox} since computing $\beta$ is slightly more efficient than
computing $\prod_{i=1}^n\frac{C_{m-i+1}}{(u_i-\ell_i+1)}$.
\end{remark}

\begin{remark}
\label{r:pbcube}
In digital communications, the box $\mathcal{B}$ is usually a n-dimensional cube.
Let $d$ be the length of the box (i.e., $d=u_i-l_i$) and $m=n$,
then \eqref{e:pbbox2} can be further rewritten as
\begin{align*}
P_e^{\sBSIC}&=1-\prod_{i=1}^n \bar{P}_{i}(d)\\
&=1-\beta\prod_{i=1}^n
\left(\frac{1}{C_{i}}+d\int_{0}^{\arctan(1/2\sigma)}\cos^{i-1}(\theta)d\theta\right),
\end{align*}
where $C_i$ is defined in \eqref{e:Ck} and
\[
\beta=\left(\frac{2}{\sqrt{\pi}(d+1)}\right)^n\frac{\Gamma((m+1)/2)}{\sqrt{\pi}}.
\]
\end{remark}

\subsection{WER Properties of BSIC Decoders}

In this subsection, we study some properties of the WER expression.
We first investigate the property of $\bar{P}_i$.
Specifically, we have the following result.
\begin{lemma}
\label{l:Rbardecreasing}
For any fixed $1\leq i \leq n$ and $\sigma$, $\bar{P}_i$ (see \eqref{e:Rbar}) is a
strictly decreasing function of $\eta$,
i.e., the following inequality holds for any $\epsilon>0$:
\beq
\label{e:Rbardecreasing}
\bar{P}_i(\eta)> \bar{P}_i(\eta+\epsilon).
\eeq
\end{lemma}
\begin{proof}
For any $1\leq i \leq n$, by \eqref{e:Rbar}, eq. \eqref{e:Rbardecreasing} is equivalent to
\begin{align*}
&\frac{C_i}{\eta+1}\left(\frac{1}{C_i}+\eta
\int_{0}^{\arctan(1/2\sigma)}\cos^{i-1}(\theta)d\theta\right)\\
>&\frac{C_i}{\eta+\epsilon+1}\left(\frac{1}{C_i}+(\eta+\epsilon)
\int_{0}^{\arctan(1/2\sigma)}\cos^{i-1}(\theta)d\theta\right).
\end{align*}
By some basic calculations, one can easily verify that the aforementioned inequality can be rewritten as
\begin{align*}
\frac{1}{C_{i}}
>\int_{0}^{\arctan(1/2\sigma)}\cos^{i-1}(\theta)d\theta.
\end{align*}
By \eqref{e:Rklimit2}, the above inequality holds.
Hence, eq. \eqref{e:Rbardecreasing} holds.
\end{proof}

By \eqref{e:pbbox} and Lemma \ref{l:Rbardecreasing}, one can easily obtain the following result.
\begin{theorem}
\label{t:BSICincrease}
Let $\mathcal{B}^1$ and $\mathcal{B}^2$ be any two $n\times n$ dimensional boxes
that satisfy $u^1_i-\ell^1_i \leq u^2_i-\ell^2_i$ for $1 \leq i \leq n$,
then the WER of BSIC decoders corresponding to $\mathcal{B}^1$ and $\mathcal{B}^2$ satisfy
\beq
\label{e:BSICincrease}
P_e^\sBSIC(\mathcal{B}^1) \leq P_e^\sBSIC(\mathcal{B}^2).
\eeq
\end{theorem}

Similar to the ordinary case, the following result holds.

\begin{theorem}
\label{t:propBSIC}
The WER $P_e^\sBSIC$ of BSIC decoders is an increasing function of $\sigma$ and $n$.
Moreover it satisfies
\[
\lim_{\sigma\rightarrow0}P_e^\sBSIC=0.
\]
\end{theorem}

\begin{proof}
Similar to the proof of Theorem \ref{t:propOSIC}, one can see that $P_e^\sBSIC$ is an increasing function of $\sigma$ and $n$.

We next  prove the second part of Theorem 6.
By \eqref{e:Rbar} and \eqref{e:integralk}, for any $1\leq i\leq n$, we have
\begin{align}
\label{e:Rkbarlimit}
&\lim_{\sigma\rightarrow0}\bar{P}_{m-i+1}(u_i-\ell_i) \nonumber\\
=&\frac{C_{m-i+1}}{u_i-\ell_i+1}\left(\frac{1}{C_{m-i+1}}+(u_i-\ell_i)\frac{1}{C_{m-i+1}}\right)=1.
\end{align}
Thus
\begin{align*}
\lim_{\sigma\rightarrow0}P_e^{\sBSIC}
=&1-\lim_{\sigma\rightarrow0}\prod_{i=1}^n\bar{P}_{m-i+1}(u_i-\ell_i)=0.
\end{align*}
Hence, the theorem holds.
\end{proof}

Note that Theorem \ref{t:propBSIC} also holds for deterministic $\A$.
For more details, see \cite[Corollary 2]{WenC17}.

Similar to OSIC decoders, for easy notation,
we denote the WER of BSIC decoders
for  $n\times n$ square matrix $\A$ and a cube $\mathcal{B}$
whose edge length is $d$ as $P_e^\sBSIC(n,d)$.
The  following results can then  be directly obtained from \eqref{e:pbbox}.

\begin{theorem}
\label{t:propBSIC2}
Let $n_1<n_2$ be two integers, then $P_e^\sBSIC(n_1,d)$ and $P_e^\sBSIC(n_2,d)$ satisfy
\beq
\label{e:SRRbox}
\frac{1-P_e^\sBSIC(n_2,d)}{1-P_e^\sBSIC(n_1,d)}=\prod_{k=n_1+1}^{n_2} \bar{P}_k(d).
\eeq
\end{theorem}

Similar to the case of OSIC, Theorem \ref{t:propBSIC2} quantifies the gap between two $P_e^\sBSIC$. Specifically, by \eqref{e:Rkbarlimit},
if $\sigma$ is close to 0, then $\bar{P}_k(d)$ is close to 1 for any integer $k$ and $d$.
Thus, eq. \eqref{e:SRRbox} indicates that when $\sigma$ is close to 0,
the difference between $1-P_e^\sBSIC(n_1,d)$ and $1-P_e^\sBSIC(n_2,d)$ is very small,
implying that the gap between $P_e^\sBSIC(n_1,d)$ and $P_e^\sBSIC(n_2,d)$ is very small
as long as noise level $\sigma$ is close to 0.
For more details, see the numerical experiments in Section \ref{s:Sim}.

\section{Relationship between $P_e^\sOSIC$ and $P_e^\sBSIC$ }
\label{s:relation}

In this section, we investigate the relationship between $P_e^\sOSIC$ and $P_e^\sBSIC$ .We first investigate the relationship between $P_i$ and $\bar{P}_i$ (see \eqref{e:R} and \eqref{e:Rbar}).
Specifically, we have the following result.
\begin{theorem}
\label{t:RRbar}
For any fixed $1\leq i \leq n$ and $\sigma$, if $\eta>0$, then $P_i$ and $\bar{P}_i$ satisfy
\beq
\label{e:RRbar}
\bar{P}_i(\eta)> P_i.
\eeq
Moreover,
\beq
\label{e:RRbar2}
\lim_{\eta \rightarrow \infty}\bar{P}_i(\eta)= P_i.
\eeq
\end{theorem}
\begin{proof}
We first show \eqref{e:RRbar}.
For any $1\leq i \leq n$, by \eqref{e:R} and \eqref{e:Rbar}, eq. \eqref{e:RRbar} is equivalent to
\begin{align*}
&\frac{1}{\eta+1}\left(\frac{1}{C_i}+\eta
\int_{0}^{\arctan(1/2\sigma)}\cos^{i-1}(\theta)d\theta\right)\\
>&\int_0^{\arctan(1/(2\sigma))}\cos^{i-1}(\theta)d\theta
\end{align*}
which can be rewritten as
\begin{align*}
\frac{1}{C_{i}}>\int_{0}^{\arctan(1/2\sigma)}\cos^{i-1}(\theta)d\theta.
\end{align*}
By \eqref{e:Rklimit2}, the above inequality holds.
Hence, eq. \eqref{e:RRbar} holds.

In the following, we prove \eqref{e:RRbar2}. Clearly, for any $1\leq i \leq n$,
\begin{align*}
\lim_{ \eta \rightarrow \infty}\bar{P}_i(\eta)=&\lim_{\eta \rightarrow \infty}\frac{C_i}{\eta+1}
\left(\frac{1}{C_i}+\eta
\int_{0}^{\arctan(1/2\sigma)}\cos^{i-1}(\theta)d\theta\right)\\
=&C_i\int_0^{\arctan(1/(2\sigma))}\cos^{i-1}(\theta)d\theta=P_i.
\end{align*}
Thus, eq. \eqref{e:RRbar2} holds.
\end{proof}

By \eqref{e:pb}, \eqref{e:pbbox} and Theorem \ref{t:RRbar}, we obtain Theorem~\ref{t:OSICBSIC},
which characterizes the relationship between $P_e^\sOSIC$ and $P_e^\sBSIC$.
\begin{theorem}
\label{t:OSICBSIC}
For any $\mathcal{B}$, $P_e^\sOSIC$ and $P_e^\sBSIC$ have the following relationship
\beq
\label{e:OSICBSIC}
P_e^\sBSIC < P_e^\sOSIC.
\eeq
Moreover,
\[
\lim_{\mbox{all} \,\;1\leq i\leq n, u_i-\ell_i \rightarrow \infty} P_e^\sBSIC = P_e^\sOSIC.
\]
\end{theorem}

Note that Theorem \ref{t:OSICBSIC} also holds for deterministic $\A$,
for more details, see \cite[Corollary 1]{WenC17}.
The inequality \eqref{e:OSICBSIC} shows that BSIC outperforms OSIC given the same level of noise.
Intuitively this is because, in OSIC, $\hat \x$ can be anywhere  in $\mathbb{Z}^n$.
In BSIC, $\hat \x$ is limited to finite number of choices,
and this property seems to improve the detection accuracy.
Theoretically, it can be showed by using \eqref{e:pb}, \eqref{e:pbbox} and Theorem \ref{t:RRbar}.

\section{Numerical Experiments}
\label{s:Sim}

We now provide simulations and numerical results  to  verify the accuracy  of the WER formulas
\eqref{e:pb} and \eqref{e:pbbox}, which are compared against  the simulated WER. Each simulation run is averaged  over $10^5$ samples.
For simplicity, we assume that $m=n$ in all of the following tests (our extensive  simulations found that both \eqref{e:pb} and \eqref{e:pbbox} are accurate for both SIC and BSIC decoders
for both $m=n$ and $m>n$).

We did the simulations by choosing a range of $n$, $\sigma$ and boxes $\mathcal{B}$
(more details on the choice of  these parameters are given subsequently).
For each fixed $n$ and $\sigma$, we randomly generated $10^5$ $\A$'s, whose entries
independent and identically follow the standard Gaussian distribution $\mathcal{N}(0,1)$,
and $10^5$ $\v$'s with each of them following the Gaussian distribution
$\mathcal{N}(\boldsymbol{0},\sigma^2 \I)$.
To illustrate the effectiveness of \eqref{e:pb},
for each generated $\A$ and $\v$, we randomly generated an $\hbx\in \mathbb{Z}^n$.
To verify the accuracy  of \eqref{e:pbbox}, for each generated $\A$ and $\v$,
we randomly generated an $\hbx$ which is uniformly distributed over a given $\mathcal{B}$.
Then, we got $10^5$ linear models which satisfy \eqref{e:model} only,
and another $10^5$ linear models which satisfy both \eqref{e:model} and \eqref{e:box}.
Then, we found $\x^\sOSIC$ and $\x^\sBSIC$ corresponding to each ordinary and box-constrained linear model according to \eqref{e:OB} and \eqref{e:BB}, respectively.
Finally, the number of events $\x^\sOSIC\neq \hbx$ divided by $10^5$
was computed as the simulated WER for OSIC decoders. Similarly, the number of events $\x^\sBSIC\neq \hbx$ divided by $10^5$
was computed as the simulated WER for BSIC decoders.
The  theoretical WERs are computed from  \eqref{e:pb} and \eqref{e:pbbox} for SIC and BSIC  decoders.

\subsection{Numerical experiments for OSIC decoders}
\label{ss:SimOSIC}
We investigate the OSIC WER to verify the accuracy  of \eqref{e:pb}.
Figure~\ref{fig:SIC} shows the  WER for several  noise standard deviations and for several sizes  $ 2 \leq n \leq 64$.
The results for  $n=64$ are added to   show  the WER of OSIC decoder for large size.
The theoretical and simulated  WERs match very well,  confirming the accuracy of  \eqref{e:pb}.
Theorem \ref{t:propOSIC} states that $P_e^\sOSIC$ increases
	when $\sigma$ or $n$ increases.
Indeed, Figure~\ref{fig:SIC} clearly demonstrates the increasing trend of
	$P_e^\sOSIC$ with noise level $\sigma$.
As size  $n$ increases, $P_e^\sOSIC$ increases slightly and then plateaus.
Although  when noise variance is small, e.g.,  high-SNR region,
$P_e^\sOSIC$ is more or less constant irrespective of size $n$.

\begin{figure}[!htbp]
\centering
\includegraphics[width=3.8 in]{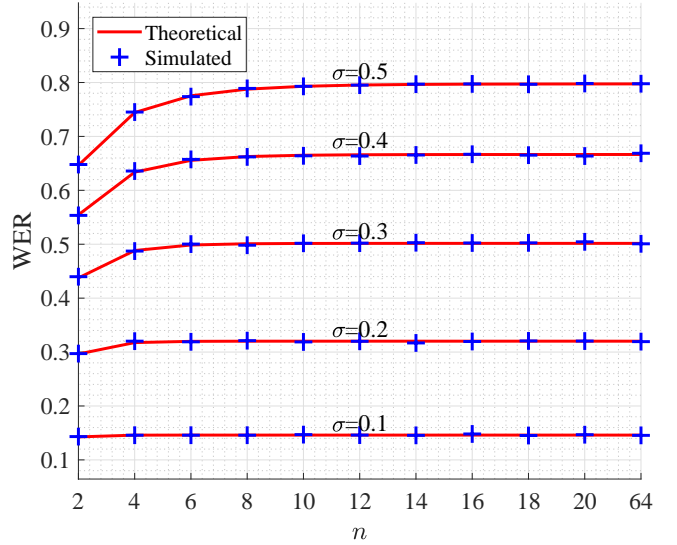}
\caption{Theoretical and simulated WER for OSIC decoders}
\label{fig:SIC}
\end{figure}

\begin{figure}[!htbp]
	\centering
	\includegraphics[width=3.8 in]{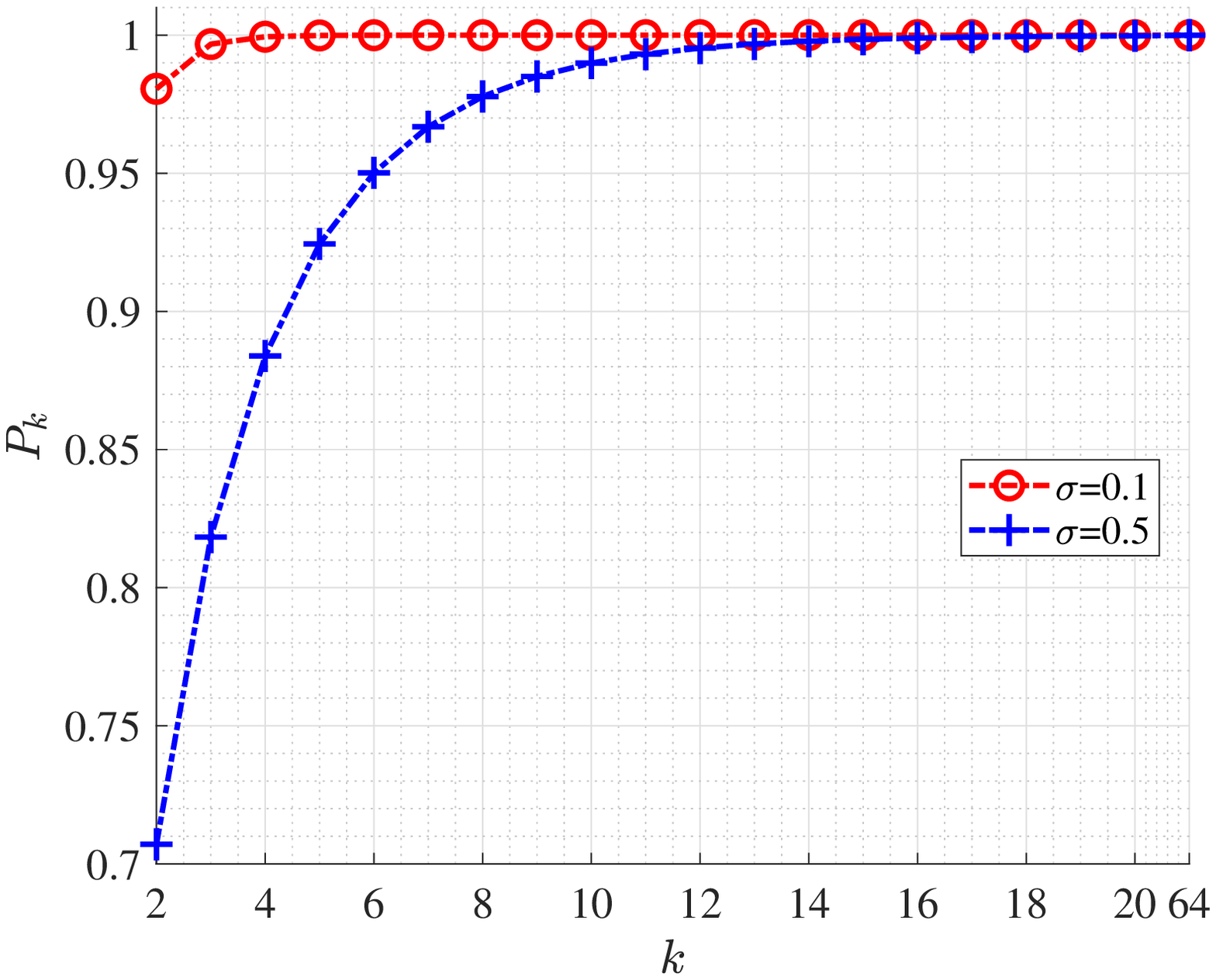}
	\caption{$P_k$ (see \eqref{e:R})} \label{fig:R}
\end{figure}

We may use Theorem \ref{t:propOSIC2} to explain the above phenomena.
The numerical  $P_k$ values  are depicted  in Figure~\ref{fig:R}
	for noise variance of  $0.1$ and $0.5$. For  both cases, $P_k$ converges to 1 as $k$ increases.
Therefore, for given $\sigma$, the performance difference between
	two OSIC detectors respectively with dimensions $n_1$ and $n_2$
	($n_2>n_1$)	is negligible, if $n_1$ is sufficient large.
Intuitively, this phenomenon is because, for OSIC,
	detection error is more likely to occur in early stages
	(see \eqref{e:SuccPron} and \eqref{e:SuccProi},
		and also notice that $P_k$ increases with $k$).
Therefore, given that all previous	stages are correctly detected,
	the probabilities of correct detection of later stages approach 1
	(notice that $P_k$ approaches 1 for sufficient large $k$).
Therefore, 	if $n$ is above a certain threshold,
	further increasing $n$ causes negligible performance deterioration.

\subsection{WER performance of BSIC decoders}
\label{ss:SimBSIC}
Here, we test the accuracy of \eqref{e:pbbox}.
Since in wireless applications  the box $\mathcal{B} $  is generally a hypercube where   $\ell_i$ and  $u_i$ are fixed and the same for $i=1, \ldots, n$. Thus, we choose
$\mathcal{B}=[0,1]^n$, $\mathcal{B}=[0,3]^n$, $\mathcal{B}=[0,7]^n$ and $\mathcal{B}=[0,63]^n$
for testing.

\begin{figure}[!htbp]
\centering
\includegraphics[width=3.8 in]{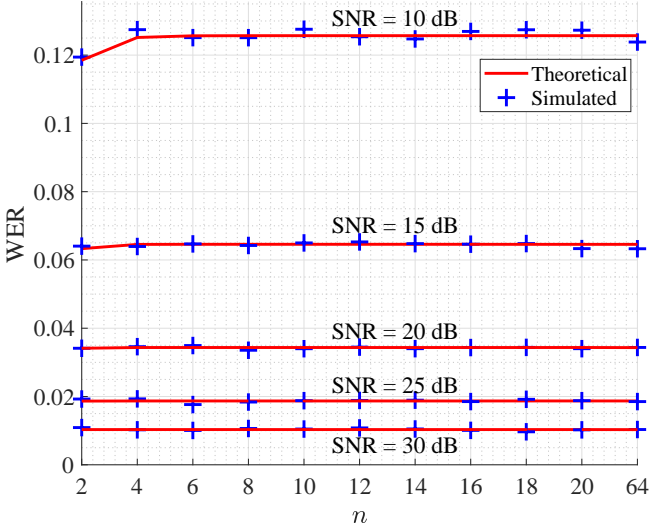}
\caption{Theoretical and simulated WER for BSIC decoders for  $\mathcal{B}=[0,1]^n$ }
\label{fig:BSIC}
\end{figure}

\begin{figure}[!htbp]
\centering
\includegraphics[width=3.8 in]{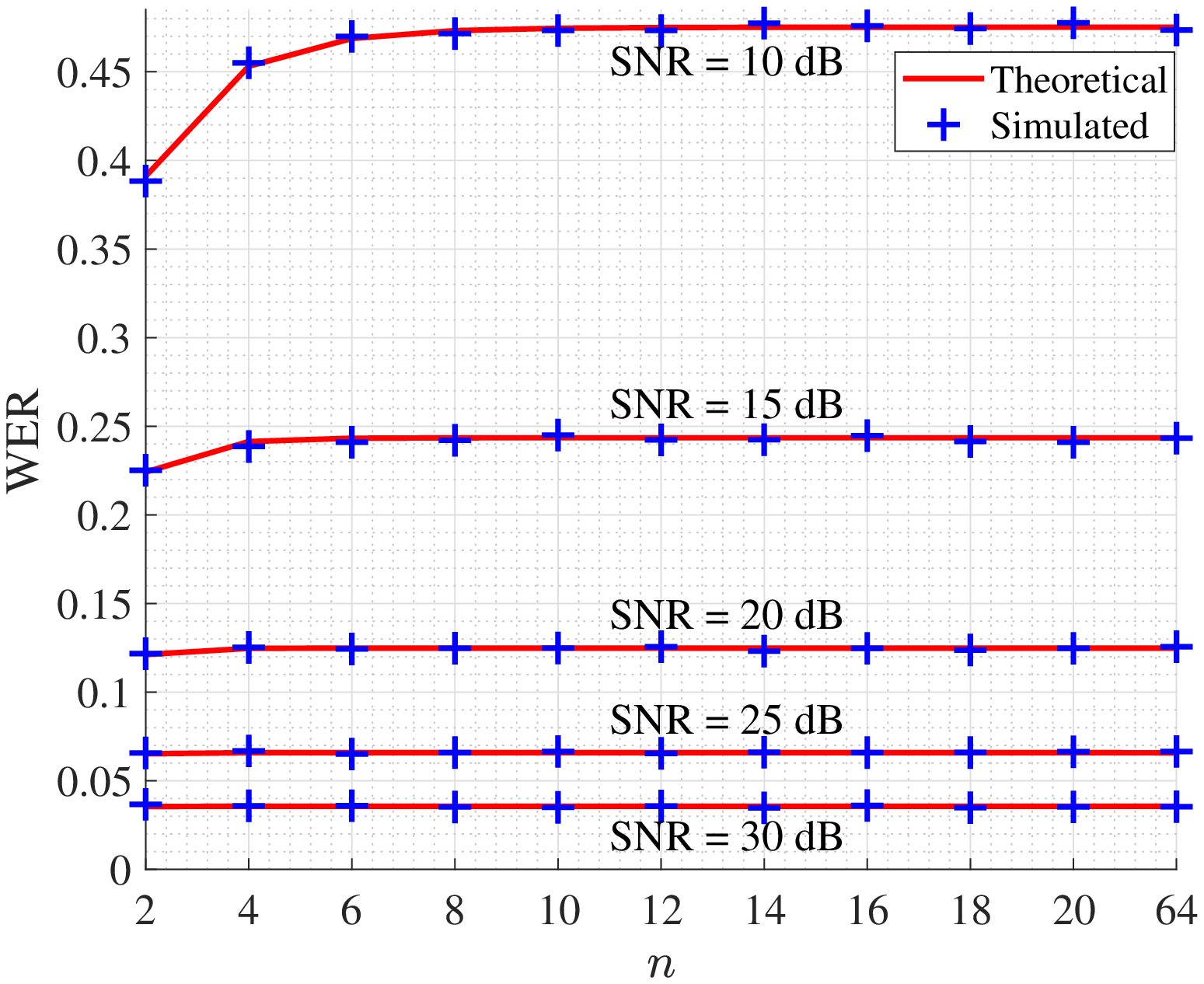}
\caption{Theoretical and simulated WER for BSIC decoders for  $\mathcal{B}=[0,3]^n$ }
\label{fig:BSIC3}
\end{figure}

For a BSIC with $\mathcal{B}=[0,u]^n$  (when $u = 2^q -1$ for some integer $q$),
each entry of $\hbx\in \mathcal{B}=[0,u]^n$
    can be viewed as a $(u+1)$-ary pulse-amplitude modulation (PAM) baseband signal\footnote{
    Strictly speaking, we have $x_i - u/2$ is equivalent to a $(u+1)$-ary baseband signal,
    	since communication signal is generally symmetric to the origin.
    }.
Furthermore, we evaluate BSIC WER in terms of signal-to-noise ratio (SNR),
	which is commonly used in wireless communications.
For a BSIC with $\mathcal{B}=[0,u]^n$, the relationship (see Appendix~\ref{app:SNR} for proof)
	between $\sigma$ and SNR in decibels (dB) is
\[
\mbox{SNR} = 10\log_{10} \frac{\mathbb E ||\hat \x||_2^2 }{n \sigma^2} = 10\log_{10}\frac{u(u+2)}{12\sigma^2}.
\]	
Figures~\ref{fig:BSIC}-\ref{fig:BSIC3} show theoretical and simulated   WER of
BSIC decoders. SNR ranges from 10 to 30 dB.
Each entry of $\hbx$ are randomly selected from 2-PAM and 4-PAM, respectively. Figures~\ref{fig:BSIC}-\ref{fig:BSIC3} show that  theoretical and  simulated  WERs match well which  confirms the accuracy  of \eqref{e:pbbox}.
It can also be observed that when the size  $n$ increases, the WER increases, which matches Theorem \ref{t:propBSIC}.
Similar to the case of OSIC, due to the decreasing error propagation nature of BSIC, the performance deterioration
caused by increasing $n$ vanishes as $n$ exceeds certain threshold (depending on SNR).

\begin{figure}[!htbp]
\centering
\includegraphics[width=3.8 in]{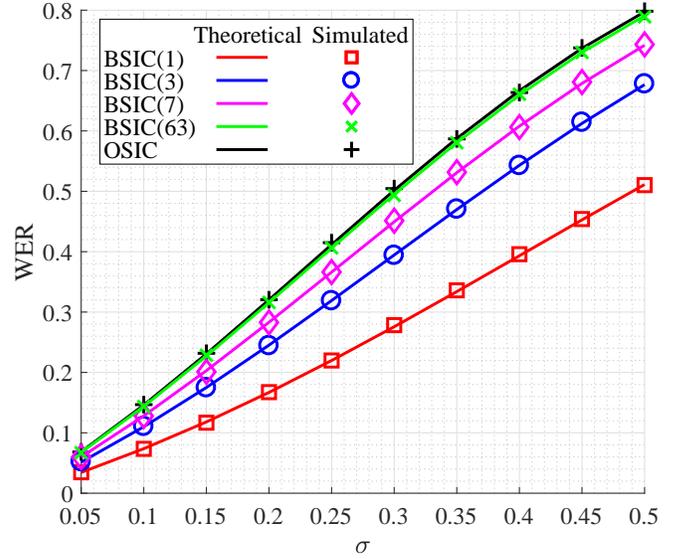}
\caption{Theoretical and simulated WER for OSIC and BSIC decoders}
\label{fig:DiffBox}
\end{figure}

Figure~\ref{fig:DiffBox} investigates WER of BSIC decoders
with  $\mathcal{B}=[0,1]^{20}$, $\mathcal{B}=[0,3]^{20}$,
$\mathcal{B}=[0,7]^{20}$ and $\mathcal{B}=[0,63]^{20}$,
denoted by BSIC(1), BSIC(3), BSIC(7) and BSIC(63), respectively.
For comparison, the OSIC with $n=20$ is also included (denoted as OSIC).
It can be recognized that, for BSIC decoders, increasing the size increases WER.
This observation matches with Theorem~\ref{t:BSICincrease}.
Furthermore, the WER of OSIC decoder exceeds that of BSIC with the same $n$ and $\sigma^2$.
Finally, when $d=63$, $P^{\sBSIC}$ appears to converge to $P^{\sOSIC}$.
Theorem \ref{t:OSICBSIC} predicts these trends.

\begin{figure}[!htbp]
	\centering
	\includegraphics[width=3.8 in]{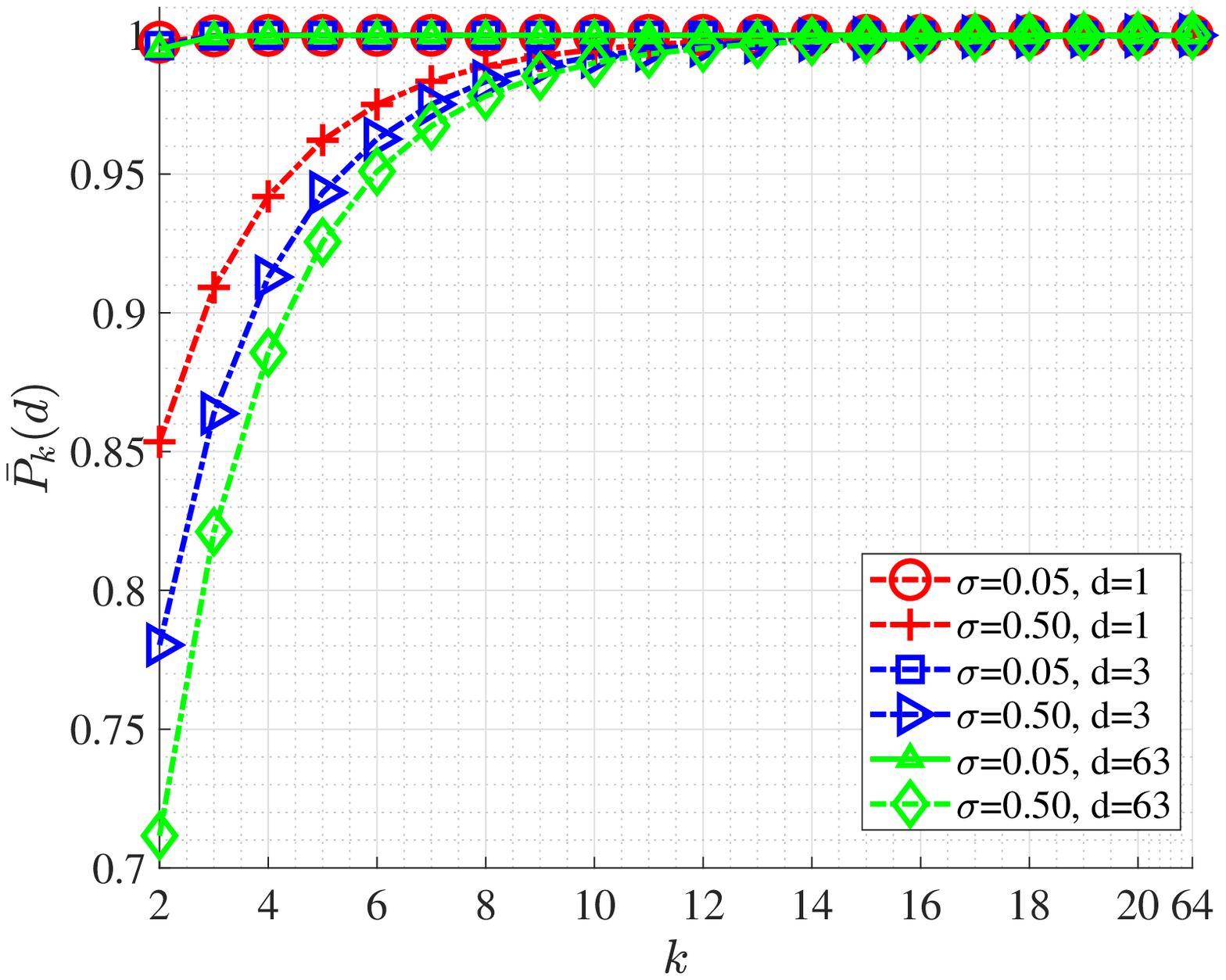}
	\caption{$\bar{P}_k(d)$ (see \eqref{e:Rbar})} \label{fig:RBox}
\end{figure}

To explain the above phenomena, we display $\bar{P}_k(d)$
under two noise levels and edge length $d=\{1,3,63\}$ in Figure~\ref{fig:RBox}.
From Figure~\ref{fig:RBox}, one can see that $\bar{P}_k(d)$ converges to 1
rapidly, especially when $\sigma=0.05$.
Also reminding that, in \eqref{e:SRRbox}, we have
\[
\frac{1-P_e^\sBSIC(n_2,d)}{1-P_e^\sBSIC(n_1,d)}=\prod_{k=n_1+1}^{n_2} \bar{P}_k(d).
\]	
Therefore, we can conclude that $P^{\sBSIC}_e(n,d)$ should change slowly for sufficiently large $n$,
	which explains the $P^{\sBSIC}_e$'s trends along $n$ in Figures~\ref{fig:BSIC}-\ref{fig:BSIC3}.
In addition, one can observe that, for given $k$ and $\sigma$, $\bar{P}_k(d)$ gets smaller when $d$ becomes larger, which is confirmed via Lemma~\ref{l:Rbardecreasing}.
This suggests that decoding performance under larger edge length decreases more with increasing $\sigma$.
Finally, by comparing Figure~\ref{fig:R} with Figure~\ref{fig:RBox},
	it can be seen that $\bar{P}_k(63)$ is very close to $P_k$,
	which is supported via Theorem~\ref{t:RRbar}.
And this explains why $P^{\sBSIC}$ with $d=63$ approaches $P^{\sOSIC}$
	in Figure~\ref{fig:DiffBox}.

\section{Summary and discussions}
\label{s:sum}

In this paper, we have derived closed-form WER expressions  $P_e^\sOSIC$ and $P_e^\sBSIC$   for  OSIC and BSIC decoders,
investigated certain  properties of the expressions  and studied their connections.
The accuracy of these expressions has been verified via simulation and numerical results.

In our model,  the entries of $\A$ are  \iid
standard Gaussian $\mathcal{N}(0,1)$ variables.   The noise vector   $\v$  follows
Gaussian distribution $\mathcal{N}(\boldsymbol{0},\sigma^2 \I)$. This model  can be readily  extended to the complex case, which is important in
 practical applications.  Thus, if  the entries of $\A$ and $\v$ are \iid
complex Gaussian, and $\hbx$ is also assumed to be a complex vector with both of its
real and image parts being uniformly distributed over a box $\mathcal{B}$.  Then just like the real case (see \eqref{e:qr}),  QR factorization of $\A$ yields  $r_{ij}, 1\leq i\leq j\leq n$, are independent,
and $r_{ii}^2\sim\chi^2_{2(m-i+1)}$ and $r_{ij}\sim \mathcal{CN}(0,1)$ for $1\leq i< j\leq n$.
One can easily obtain formulas for $P_e^\sOSIC$  and $P_e^\sBSIC$
under complex $\A$, $\hbx$ and $\v$ by using  the  techniques developed in this paper. Thus, we omit the details.

Theoretical results \cite{ChaWX13} show that the LLL reduction can always
decrease (not strictly) $P_e^\sOSIC$ for deterministic $\A$.
It is straightforward to see that the LLL reduction can also always
decrease (not strictly) $P_e^\sOSIC$ for random $\A$.
Thus, it is important to develop a formula for
$P_e^\sOSIC$ after the LLL reduction is performed on $\A$.
But to do this, we need to find the distribution of the entries of $\bbR$, which is the
LLL reduced matrix of $\R$ (see \eqref{e:qr}).
However, to the best of our knowledge, this is still an open problem due to the complication
of the LLL reduction.

It is well-known that some of the permutation strategies, such as V-BLAST \cite{FosGVW99} and SQRD \cite{WubBRKK01}, can usually decrease $P_e^\sBSIC$ for deterministic $\A$.
This property  also holds for random $\A$.
Thus,  closed-form  $P_e^\sBSIC$ when  $\A$ is column permuted may be useful, which is a potential future research problem.
In addition to these traditional detection strategies, one can also use a naive
lattice decoder \cite{TahK10} to detect $\hbx$ (e.g. perform  traditional lattice decoding and  discard the vectors  not in
the box $\mathcal{B}$ \cite{TahK10}).
The naive lattice decoder performs better for (4)  than for the  ordinary linear model.
Furthermore,  naive lattice decoding achieves maximum diversity \cite{TahMK07}. Since this decoding  is complicated, closed-form analysis of its  WER appears intractable.

On the other hand, although the LLL reduction algorithm  reduces $n$-dimensional lattices, whose basis vectors are integer vectors, in polynomial time of $n$ (see \cite{LenLL82}, \cite{DauV94}),
and the average complexity of reducing an \iid Gaussian  matrix $\A$  is also a polynomial of the column rank of $\A$ (\cite{LinMH13}, \cite{JalSM08}),
the worst-case complexity of LLL is not even finite \cite{JalSM08}.
This suggests a potential use for closed-form  $P_e^\sOSIC$. For instance,  if $P_e^\sOSIC$  is  smaller than a suitable threshold, we may not employ
 LLL reduction; thus, in  practical applications, LLL reduction may be applied adaptively.
Similarly,  closed-form  $P_e^\sBSIC$ can be useful.

Minimum mean square error (MMSE) decoder is a popular  alternative to OSIC and BSIC decoders.
MMSE decoder adapts to the noise level  \cite{WueBKK04}. A closed-form WER  of  MMSE is a potential  future research topic.

\appendices
\section{Proof of Lemma~\ref{l:probk}}\label{s:probkproof}

\begin{proof}
By \eqref{e:modelk},
\[
x=\lfloor\bar{y}/r\rceil=\lfloor\hat{x}+\bar{v}/r\rceil
=\hat{x}+\lfloor\bar{v}/r\rceil,
\]
thus, $x=\hat{x}$ if and only if $|\bar{v}/r|\leq 1/2$.

Let $X=\bar{v}^2$, $Y=r^2$ and  $U=X/Y$. Thus,  $x=\hat{x}$ if and only if $U\leq 1/4$.
Thus, to show \eqref{e:R}, we derive $\Pr\left(U\leq 1/4\right)$.  Note that $U$ is the  ratio of  two independent  central chi-square  random variables. The distribution of this  ratio   is well-known \cite[Section 27]{johnson}. That is,  $ U=  \frac {\sigma^2} k  \frac {\chi_1^2 }{\chi_k^2/k} =  \frac {\sigma^2} k F_{1,k} $ where $F_{1,k}$ an $F$ distributed rv.  Thus, the PDF of $F_{1,k}$ is given by
\[ f_{1,k}(x) = \frac{\Gamma(\frac{1+k}{2})k^{k/2}}{\Gamma(\frac{1}{2})\Gamma(\frac{k}{2})} \frac{x^{-1/2}}{(k+x)^{(k+1)/2}}, \quad x \geq 0. \] Therefore, we find
\begin{align}
\label{e:p-k2}
\Pr(U\leq \frac{1}{4})=&\int_0^{k/4\sigma^2 } f_{1,k}(x)dx \nonumber\\
= & \int_0^{k/4\sigma^2 } \frac{\Gamma(\frac{1+k}{2})k^{k/2}}{\Gamma(\frac{1}{2})\Gamma(\frac{k}{2})} \frac{x^{-1/2}}{(k+x)^{(k+1)/2}} dx  \nonumber\\
=&C_k\int_0^{\arctan(1/(2\sigma))}\cos^{k-1}(\theta)d\theta,
\end{align}
where the last equality follows from the substitution  $x=k\tan^2(\theta)$.
Thus, the lemma holds.
\end{proof}

\section{Proof of Lemma~\ref{l:probkbox2}}
\label{s:probkproofbox2}

\begin{proof}
Since $\hat{x}$ is uniformly distributed on $[\ell,u]$, we have
\begin{align}
\label{e:pksum}
&\Pr(x=\hat{x})\nonumber\\
=&\Pr((x=\hat{x})\cap(\hat{x}=\ell))+\Pr((x=\hat{x})\cap(\hat{x}=u))\nonumber\\
&+\Pr((x=\hat{x})\cap(\ell<\hat{x}<u))\nonumber\\
=&\Pr(x=\hat{x}|\hat{x}=\ell)\Pr(\hat{x}=\ell)\nonumber\\
&+\Pr(x=\hat{x}|\hat{x}=u)\Pr(\hat{x}=u)\nonumber\\
&+\Pr(x=\hat{x}|\ell<\hat{x}<u)\Pr(\ell<\hat{x}<u)\nonumber\\
=&\frac{1}{u-\ell+1}[\Pr(x=\hat{x}|\hat{x}=\ell)+\Pr(x=\hat{x}|\hat{x}=u)\nonumber\\
&\quad\quad\quad\quad+(u-\ell-1)\Pr(x=\hat{x}|\ell<\hat{x}<u)].
\end{align}
In the following, we derive formulas for
\[
\Pr(x=\hat{x}|\hat{x}=\ell),\,\Pr(x=\hat{x}|\hat{x}=u)\;\text{and} \,\Pr(x=\hat{x}|\ell<\hat{x}<u).
\]

Let $W=\bar{v}/r$, then by \eqref{e:modelk},
$
\lfloor\bar{y}/r\rceil=\lfloor\hat{x}+\bar{v}/r\rceil=\hat{x}+\lfloor W\rceil.
$
From \eqref{e:BBk}, we can see that
\beqnn
\begin{split}
 x=
\begin{cases}
\ell, & \mbox{ if }\   \hat{x}+\lfloor W\rceil\leq \ell\\
\hat{x}+\lfloor W\rceil, & \mbox{ if }\    \ell<\hat{x}+\lfloor W\rceil< u\\
u, & \mbox{ if }\    \hat{x}+\lfloor W\rceil\geq u
\end{cases}.
\end{split}
\eeqnn
Thus, $x=\hat{x}$ if and only if
\beqnn
\begin{split}
 W\in
\begin{cases}
(-\infty,1/2], & \mbox{ if }\   \hat{x}=\ell\\
[-1/2,1/2], & \mbox{ if }\    \ell<\hat{x}<u\\
[-1/2,+\infty), & \mbox{ if }\    \hat{x}=u
\end{cases}.
\end{split}
\eeqnn

We first show how to compute $\Pr(x=\hat{x}|\hat{x}=\ell)$.
Since $\bar{\v}$ and $r^2$ are independent, by the distribution of $\bar{v}$ and $r^2$, we can see that
the PDF of $W$ is symmetric with $x=0$. Thus,
\begin{align}
&\Pr(x=\hat{x}|\hat{x}=\ell)\nonumber\\
=&\Pr(W\leq1/2)\nonumber\\
=&\Pr(W<0)+\Pr(0\leq W\leq1/2)\nonumber\\
=&\frac{1}{2}\big(1+\Pr(-1/2\leq W\leq1/2)\big)\nonumber\\
\overset{(a)}{=}&\frac{1}{2}\big(1+\Pr(U\leq1/4)\big)\nonumber\\
\overset{(b)}{=}&\frac{1}{2}\left(1+
C_k\int_{0}^{\arctan(1/2\sigma)}\cos^{k-1}(\theta)d\theta\right)\nonumber\\
\overset{(c)}{=}&\frac{1}{2}\left(C_k\int_{0}^{\pi/2}\cos^{k-1}(\theta)d\theta\right.\nonumber\\
&\quad\left.+C_k\int_{0}^{\arctan(1/2\sigma)}\cos^{k-1}(\theta)d\theta\right)\nonumber\\
=&\frac{C_k}{2}\left(\int^{0}_{-\pi/2}\cos^{k-1}(\theta)d\theta
+\int_{0}^{\arctan(1/2\sigma)}\cos^{k-1}(\theta)d\theta\right)\nonumber\\
=&\frac{C_k}{2}\int_{-\pi/2}^{\arctan(1/2\sigma)}\cos^{k-1}(\theta)d\theta,
\label{e:Case I}
\end{align}
where $(a)$ is because $U=W^2$, (b) follows from \eqref{e:p-k2}
and (c) is from \eqref{e:integralk}.

Similarly, we have
\begin{align}
\label{e:Case II}
&\Pr(x=\hat{x}|\ell<\hat{x}<u)\nonumber\\
=&\Pr(-1/2\leq W\leq1/2)=\Pr(U\leq1/4)\nonumber\\
=&C_k\int_{0}^{\arctan(1/2\sigma)}\cos^{k-1}(\theta)d\theta.
\end{align}

Since the PDF of $W$ is symmetric with $x=0$, we have
\begin{align}
\Pr(x=\hat{x}|\hat{x}=u)
=&\Pr(W\geq-1/2)=\Pr(W\leq1/2)\nonumber\\
=&\frac{C_k}{2}\int_{-\pi/2}^{\arctan(1/2\sigma)}\cos^{k-1}(\theta)d\theta.
\label{e:Case III}
\end{align}

Then, by \eqref{e:pksum}-\eqref{e:Case III}, we have

\begin{align*}
&\Pr(x=\hat{x})\nonumber\\
=&\frac{C_k}{u-\ell+1}\left(\int_{-\pi/2}^{\arctan(1/2\sigma)}\cos^{k-1}(\theta)d\theta\right.
\nonumber\\
&\quad+\left.(u-\ell-1)\int_{0}^{\arctan(1/2\sigma)}\cos^{k-1}(\theta)d\theta\right)\nonumber\\
=&\frac{C_k}{u-\ell+1}\left(\int_{-\pi/2}^{0}\cos^{k-1}(\theta)d\theta\right.
\nonumber\\
&\quad+\left.(u-\ell)\int_{0}^{\arctan(1/2\sigma)}\cos^{k-1}(\theta)d\theta\right)\nonumber\\
=&\frac{C_k}{u-\ell+1}\left(\frac{1}{C_k}
+(u-\ell)\int_{0}^{\arctan(1/2\sigma)}\cos^{k-1}(\theta)d\theta\right),
\end{align*}
where the last equality is from \eqref{e:integralk}.
Thus, by \eqref{e:Rbar}, eq. \eqref{e:pkbox2} holds.
\end{proof}

\section{Derivation of SNR} \label{app:SNR}
In the following, we give the relationship between SNR in dB and
$\sigma$ for the case that $\hat{\x}$ is uniformly distributed in a
box $\mathcal{B} = [0,u]^n$ ($u=2^q-1$ for some integer $q$), which is transformed from an $n$-dimensional $(u+1)$-ary PAM.
Specifically, for any signal $\bar{\x}$ in an $n$-dimensional $(u+1)$-ary PAM, i.e., $\bar{x}_i\in\{-\frac{u}{2},-\frac{u-2}{2},\cdots,\frac{u-2}{2},\frac{u}{2}\}$, we let
$\hat{\x}=\bar{\x}+u/2\e$, where $\e$ is an $n$-dimensional vector
 with all of its entries being 1, then $\hat{\x}\in \mathcal{B} = [0,u]^n$.

Since $\mathcal{B} = [0,u]^n$ is transformed from an $n$-dimensional $(u+1)$-ary PAM, we calculate $\mathbb E \|\bar{\x}\|_2^2$ over the
$n$-dimensional $(u+1)$-ary PAM instead of $\mathbb E \|\hat{\x}\|_2^2$ over $\mathcal{B}$.
Since each entry of $\bar{\x}$ belongs to a $(u+1)$-ary PAM,
there are $(u+1)^n$ number of different $\bar{\x}$, and hence
\beq
\label{e:Ex}
\mathbb E \|\bar{\x}\|_2^2 =\frac{1}{(u+1)^n} \sum_{\bar{\x}\in \,\, n-\mbox{dimensional} \, \;(u+1)-\mbox{ary PAM}}
\|\bar{\x}\|_2^2.
\eeq
Each $\bar{\x}$ has $n$ entries, so the total number of
entries of all the different $\bar{\x}'s$ are $n(u+1)^n$.
Since $\bar{\x}$ is uniformly distributed over $n$-dimensional $(u+1)$-ary PAM, each entry of $\bar{\x}$ is also
uniformly distributed over $(u+1)$-ary PAM, which implies that
each point in the $(u+1)$-ary PAM are chosen
\[
\frac{n(u+1)^n}{u+1}=n(u+1)^{n-1}
\]
times.
Therefore,
\begin{align*}
&\sum_{\bar{\x}\in n-\mbox{dimensional} \, \;(u+1)-\mbox{ary PAM}}
\|\bar{\x}\|_2^2\\
=&n(u+1)^{n-1}\\
&\times\left[ (-\frac{u}{2})^2 + (-\frac{u-2}{2})^2+\cdots + (\frac{u-2}{2})^2+ (\frac{u}{2})^2 \right]\\
=&\frac{ n (u+1)^{n-1}(u+1)((u+1)^2-1) }{12 }\\
=&\frac{ n (u+1)^{n}((u+1)^2-1) }{12 }=\frac{ n (u+1)^{n}u(u+2) }{12 }.
\end{align*}
Then by \eqref{e:Ex}, we have
\[
\mathbb E \|\bar{\x}\|_2^2 =\frac{n u(u+2)}{12}.
\]
Therefore, we SNR in dB satisfies
\begin{align*}
\text{SNR} =10 \log_{10}\frac{\|\bar{\x}\|_2^2}{n\sigma^2}
%=10 \log_{10} \frac{ n ((u+1)^2-1) }{12 n \sigma^2}
=  10 \log_{10} \frac{  u(u+2) }{12 \sigma^2}.
\end{align*}

\bibliographystyle{IEEEtran}
\bibliography{ref}

\end{document}